\documentclass[twoside,11pt]{article}

\usepackage{blindtext}

\usepackage[preprint]{jmlr2e}

\usepackage{lscape} 
\usepackage{varioref}
\usepackage{multirow}
\usepackage{booktabs}
\usepackage[flushleft]{threeparttable}
\usepackage{amsmath}
\usepackage[capitalise,noabbrev]{cleveref}
\usepackage{float}
\usepackage{comment}
\usepackage{soul}
\usepackage{algorithm}
\usepackage[noend]{algpseudocode}
\newcommand{\vect}[1]{{\boldsymbol{#1}}} % boldface vectors
\newcommand{\lo}{\text{low}} %independent
\newcommand{\hi}{\text{high}} %independent
\usepackage{lastpage}
\ShortHeadings{Conformal Hyperrectangles}{Sampson and Chan}
\firstpageno{1}

\begin{document}

\title{Conformal Multi-Target Hyperrectangles}

\author{\name Max Sampson \email max-sampson@uiowa.edu \\
       \addr Department of Statistics \& Actuarial Science\\
       University of Iowa\\
       Iowa City, Iowa 52242-1409, USA
       \AND
       \name Kung-Sik Chan \email kung-sik-chan@uiowa.edu \\
       \addr Department of Statistics \& Actuarial Science\\
       University of Iowa\\
       Iowa City, Iowa 52242-1409, USA}

\editor{TBD}

\maketitle

\begin{abstract}%   <- trailing '%' for backward compatibility of .sty file

We propose conformal hyperrectangular prediction regions for multi-target regression. We propose split conformal prediction algorithms for both point and quantile regression to form hyperrectangular prediction regions, which allow for easy marginal interpretation and do not require covariance estimation. In practice,  it is preferable that a prediction region is balanced, that is, having identical marginal prediction coverage, since prediction accuracy is generally equally important across components of the response vector. The proposed algorithms possess two desirable properties, namely, tight asymptotic overall nominal coverage as well as asymptotic balance, that is, identical asymptotic marginal coverage, under mild conditions. We then compare our methods to some existing methods on both simulated and real data sets. Our simulation results and real data analysis show that our methods outperform existing methods while achieving the desired nominal coverage and good balance between dimensions.

\end{abstract}


\author{Max Sampson \& Kung-Sik Chan}

\section{Introduction}

% \subsection{Setup}
% \sloppy stops the line from overflowing
\sloppy We consider the problem of forming semi-parametric prediction regions with a $p$-dimensional response for $p \ge  2$. Assume  $n$ i.i.d. (more generally, exchangeable) copies of $((\vect{Y}_1, \vect{X}_1), (\vect{Y}_2, \vect{X}_2)\ldots, (\vect{Y}_n, \vect{X}_n))$, are observed. Then, we want a set, $\vect{C}(\vect{x}) = \vect{C}_n((\vect{Y}_1, \vect{X}_1),$ $(\vect{Y}_2, \vect{X}_2),\ldots,$ $(\vect{Y}_n, \vect{X}_n), \vect{x})$ such that for a new pair $(\vect{Y}_{n+1}, \vect{X}_{n+1})$,
\begin{equation}
\mathbb{P}(\vect{Y}_{n+1} \in \vect{C}(\vect{X}_{n+1})) \geq 1-\alpha
\label{eq:p_dim_response_coverage}
\end{equation}
If a suitable non-conformity score is used, this coverage can be guaranteed in finite samples using conformal prediction \citep{conformal_book, dis_free_pred_COPS}. As with other conformal prediction methods, we aim to make no distributional assumptions with our prediction sets. Because of that, we focus on a hyperrectangular shape for the prediction sets. This is a similar idea to hyperrectangular confidence intervals and balanced bootstrap intervals from \cite{robertson_jones_2021_hyperrectangle, beran_1998_balanced_boostrap, romano_balanced_rectangle_2010}. One advantage that hyperrectangular regions have compared to ellipsoidal regions is that they are easier to visualize in more than two dimensions. Along with being easier to visualize, a marginal interval can also easily be found. For example, in a two dimensional response case, the marginal interval for the first response is the same regardless of the second response's value. This would not be the case with an ellipsoidal prediction region. Though, this marginal interval may have a coverage rate higher than $1 - \alpha$. Another advantage is that no estimate of a covariance structure between dimensions needs to be made for rectangular regions. 
% This allows our method to be completely non-parametric.

Borrowing from conformal prediction for multivariate functional data, one option would be to use a non-conformity score of $V = \max\limits_{i = 1, \ldots, q}|Y_i - \hat{Y}_i|$ \citep{diquigiovanni2021_functional_conformal}. This guarantees the coverage in \eqref{eq:p_dim_response_coverage}, but fails to take into account the fact that one (or more) of the $p$-dimensions may have a higher variability than the others, causing overcoverage and poor balance. There are also methods proposed in \cite{kuleshov2018a_pdim, johnstone2021_pdim} that use distance functions, like the L2 norm or Mahalanobis distance, for their non-conformity scores. While in theory, this is a valid approach, it requires a $p-$dimensional grid search which limits the practical use to a few dimensions. The Mahalanobis distance also requires a covariance matrix to be estimated. Other methods use a Bonferroni or similar correction for independently modeled responses \citep{neeven2018a_pdim, messoudi2020a_pdim}. This can cause overcoverage and wider than necessary regions. We outline a new non-conformity score that attempts to take the differences in variability into account, while guaranteeing the coverage in \eqref{eq:p_dim_response_coverage}. Moreover, the new non-conformity score is shown to achieve an asymptotically balanced conformal prediction hyperrectangle under mild conditions. %; by perfect balance, it is meant that the marginal prediction intervals of the conformal hyperrectangle enjoy identical coverage rates.   

For multiple comparisons, balance is defined as having the same type 1 error rate for each comparison \citep{beran_1998_balanced_boostrap, tukey_comparisons_1953, scheffe_contrasts_1953, roy_bose_balance_1953}. We extend this idea to hyperrectangular prediction regions for a $p$-dimensional response. Balance for hyperrectangular prediction regions means that the marginal prediction intervals enjoy identical coverage rates. For multiple comparisons, balance has a natural interpretation, fairness. For example, the importance of one follow-up test from an ANOVA model is not more important than another. One wouldn't conduct follow-up tests where the first comparison has a marginal type 1 error of $\alpha$ and all others have a marginal type 1 error of $0$. The same is true for a $p$-dimensional prediction region, coverage in one dimension is not less important than coverage in another. 

We can imagine a scenario where we construct a joint prediction region such that the overall coverage rate is $1 - \alpha$. The first $k - 1$ dimensions have marginal prediction intervals of $\mathbb{R}$, and the last has a marginal prediction set that has a coverage rate equal to $1 - \alpha$. This would be a practically useless prediction region. But, we can imagine a more realistic scenario where we disregard balance and have $k - 1$ prediction sets that have marginal coverage slightly less than $1$, and one prediction set that has marginal coverage slightly larger than $1 - \alpha$. Clearly, this would be undesirable from a practical standpoint \citep{Bruder_2017_Balanced}. 

We consider point regression, including mean regression, in Section~\ref{sec:mean_rectangle}, and detail our proposed algorithm in~\cref{alg:p_dim_mean}. A drawback of the point regression approach is that the conformal hyperrectangles with a fixed coverage rate are of identical (hyper)volume for any set of covariates. For conditionally heteroscedastic data, we would want the volume of the hyperrectangle to vary based on the observed covariates. We consider the use of quantile regression for handling conditional heteroscedascity in Section~\ref{sec:quant_rectangle}, and detail the algorithm in~\cref{alg:p_dim_quant}. For ease of exposition, we assume that $p$ = 3 in both cases, but it is easy to see how it could be extended to dimensions higher than 3. We discuss the validity and the balance of the proposed methods in Sections~\ref{sec:validity} and~\ref{sec:balance}. We illustrate the empirical performance of the proposed methods via simulations in Section~\ref{sec:simulation}, and a real application in Section~\ref{sec:real}. We conclude in Section~\ref{sec:conclusion}. 

\section{Conditional Point Regression} \label{sec:mean_rectangle}

We now describe conformal hyperrectangular regression (CHR). 
% \marginpar{Should this be conformal hyperrectangular regression or conformal rectangular regression? Same question for the quantile version}. 
As with split conformal prediction, we begin by splitting our data into sets used for training and calibration \citep{split_conformal_lei_2016, conformal_book}. However, unlike classical split conformal prediction, we split using two calibration sets instead of one. The training set is indexed by $\mathcal{I}_{tr}$, the first calibration set by $\mathcal{I}_{cal1}$, and the second calibration set by $\mathcal{I}_{cal2}$. Given any point regression algorithm, $\vect{f}$, we fit a $p-dimensional$ function $\hat{\vect{f}}$ on the training set:
\[
\{\hat{\vect{f}}\} \leftarrow \vect{f}(\{(\vect{Y}_i, \vect{X}_i): i \in \mathcal{I}_{tr}\}).
\]

In the next step, for each of the $p$-dimensions of the response, we compute non-conformity scores. These scores will depend on the non-conformity measure one wishes to use; we'll use the absolute difference to demonstrate how to form CHR prediction regions. The scores are computed on the first calibration set:
\[
V_{i,j} = |\hat{f}_{j}(\vect{X}_{i,j}) - Y_{i,j}|, \text{ } \forall i \in \mathcal{I}_{cal1}.
\]
Define the vectors $\vect{V_l} = (V_{1l}, V_{2l}, \ldots, V_{n{cal_1}l}), \text{ } l = 1, \ldots, p$. Then, we form initial intervals based on the second calibration set for each $j = 1, \ldots, p$. The $p$th dimension's intervals take the form
\begin{equation}
C_p(\vect{X}_{k}) = [\hat{f}_p(\vect{X}_{k}) - Q_{1-\alpha}(\vect{V}_p; Z_{cal1}), \hat{f}_p(\vect{X}_{k}) + Q_{1-\alpha}(\vect{V}_p; Z_{cal1})], \text{ } \forall k \in \mathcal{I}_{cal2}
\label{eq:p_dim_mean_rect_formation}
\end{equation}
where,
\[
Q_{\delta}(\vect{V}; Z_{cal}) := (\delta)(1 + \frac{1}{|Z_{cal}|})-\text{th empirical quantile of }\{V_i\}
\]
The other $p-1$ prediction intervals take the same form with their corresponding non-conformity scores. Define the lower-bound of the prediction interval in \eqref{eq:p_dim_mean_rect_formation} as $\hat{q}_{j, \lo}(\vect{X}_k)$ and the upper-bound as $\hat{q}_{j, \hi}(\vect{X}_k)$. 

At this stage, in an effort to have some balance between the dimensions, we record the length of the $p$ sides. Let $A$ denote the length of the first interval, $B$ denote the length of the second interval, and $C$ the length of the third interval. If we use the absolute non-conformity score or the signed error non-conformity score, these lengths will be constant for all of the $ncal_1$ intervals. How to deal with non-conformity scores that give different sized prediction intervals for difference covariates is covered in~\cref{sec:quant_rectangle}.

Now, using the intervals found in \eqref{eq:p_dim_mean_rect_formation}, we compute a non-conformity score to make sure our rectangular prediction regions will have the coverage described in \eqref{eq:p_dim_response_coverage}. For each of the $p$-dimensions of the response, compute the following non-conformity scores:
\[
E_{k,j} = \max\{\hat{q}_{j, \lo}(\vect{X}_k) - Y_{k,j}, Y_{k,j} - \hat{q}_{j, \hi}(\vect{X}_k)\}, \ \forall k \in \mathcal{I}_{cal2} \text{ and } j = 1, \ldots, p.
\]
We could just use these as the final non-conformity scores, but to help with balance and reducing the size of the overall prediction regions, we now convert $E_{k,j}$ to the length of the first dimension using the side lengths we found earlier. Denote these converted scores as $A_{ik}$. For clarity, some examples of conversions would be $A_{11} = E_{11}$, $A_{12} = E_{12} \times \frac{A}{B}$, $A_{13} = E_{13} \times \frac{A}{C}$, etc. Now define $W_k = \max\limits_{j = 1, \ldots, p} A_{k,j}$. 

Our final steps are to find the adjustments for each of the dimensions, and construct the final prediction region. Define $Adj_1$ to be the  $(1 - \alpha)(1 + \frac{1}{|Z_{cal2}|})-$th empirical quantile of $\{W_k\}$. This is the final adjustment for the first (more generally, the reference) dimension. We can find the adjustments for the other dimensions by undoing the conversion to the reference side. For clarity, some examples of this with the first dimension as the reference would be, $Adj_2 = Adj_1 \times \frac{B}{A}$, $Adj_3 = Adj_1 \times \frac{C}{A}$, etc. Finally, given new input data, $\vect{X}_{n+1}$, the prediction region for $\vect{Y}_{n+1}$ is constructed as 
\[
\hat{\vect{C}}(\vect{X}_{n+1}) = [\hat{\vect{q}}_{\lo}(\vect{X}_{n+1}) - \vect{Adj}, \hat{\vect{q}}_{\hi}(\vect{X}_{n+1}) + \vect{Adj})].
\]
For reference, the procedure is summarized in \cref{alg:p_dim_mean}.

\begin{algorithm}
    \caption{P-dimensional Point Conformal Regression}\label{alg:p_dim_mean}
    \textbf{Input:} level $\alpha$, data = $\mathcal{Z} = (\vect{Y}_i, \vect{X}_i)_{i \in \mathcal{I}}$, score function $V$, test point $\vect{x}$, and regression algorithm $\vect{f}(\vect{X}; \mathcal{D})$ for the point estimate using $\mathcal{D}$ as data \newline
    \textbf{Procedure:}
    \begin{algorithmic}[1]

    \State Split $\mathcal{Z}$ into a training fold $\mathcal{Z}_{tr} \overset{\Delta}{=} (\vect{Y}_i, \vect{X}_i)_{i \in \mathcal{I}_{tr}}$ and two calibration folds $\mathcal{Z}_{cal1} \overset{\Delta}{=} (\vect{Y}_i, \vect{X}_i)_{i \in \mathcal{I}_{cal1}}$ and $\mathcal{Z}_{cal2} \overset{\Delta}{=} (\vect{Y}_i, \vect{X}_i)_{i \in \mathcal{I}_{cal2}}$
    \State Train the model $\vect{\hat{f}}(\vect{x}; Z_{tr})$ using the training set
    \State For each $i \in \mathcal{I}_{cal1}$ and each $j \in 1, \ldots, p$, compute the score $V_{i,j}$
    \State Find $p \times |\mathcal{Z}_{cal2}|$ $1 - \alpha$ conformal prediction intervals using the second calibration set's covariates. Define these intervals as $(\hat{q}_{j, \lo}(\vect{X}_k), \hat{q}_{j, \hi}(\vect{X}_k))$ for each $k \in \mathcal{I}_{cal2}$
    \State Find the $p - 1$ ratios of these intervals and define them as $\frac{A}{B} = \frac{\hat{q}_{11, \hi}(\vect{X}_1) - \hat{q}_{11, \lo}(\vect{X}_1)}{\hat{q}_{12, \hi}(\vect{X}_1) - \hat{q}_{12, \lo}(\vect{X}_1)}$, $\frac{A}{C} = \frac{\hat{q}_{11, \hi}(\vect{X}_1) - \hat{q}_{11, \lo}(\vect{X}_1)}{\hat{q}_{13, \hi}(\vect{X}_1) - \hat{q}_{13, \lo}(\vect{X}_1)}$, \ldots
    \State For each $k \in \mathcal{I}_{cal2}$ and each $j \in 1, \ldots, p$ compute the score $E_{k,j} = \max\{\hat{q}_{j, \lo}(\vect{X}_k) - Y_{k,j}, Y_{k,j} - \hat{q}_{j, \hi}(\vect{X}_k)\}$
    \State For each $k \in \mathcal{I}_{cal2}$ convert the $E_{k,j}$ to the first dimension using the ratios found in step 5. For example, $A_{11} = E_{11}$ and $A_{12} = E_{12} * \frac{A}{B}$
    \State Define $W_k = \max_{j} A_{k,j}$
    \State Find $Adj_1 = (1 - \alpha)(1 + \frac{1}{|Z_{cal2}|})-\text{th empirical quantile of }\{W_k\}$  
    \State Using the ratios from step 5, find $Adj_2 = Adj_1 * \frac{B}{A}$, $Adj_3 = Adj_1 * \frac{C}{A}$, \ldots
    \end{algorithmic}
    \textbf{Output:} $\hat{\vect{C}}(\vect{x}) = (\hat{\vect{q}}_{\lo}(\vect{x}) - \vect{Adj}, \hat{\vect{q}}_{\hi}(\vect{x}) + \vect{Adj})$, where $\hat{\vect{q}}_{\lo}(\vect{x})$ and $\hat{\vect{q}}_{\hi}(\vect{x})$ are the p-dimensional prediction sets for $\vect{x}$ defined in step 4 and $\vect{Adj}$ is the p-dimensional vector of adjustments found in steps 9 and 10

\end{algorithm}

\subsection{Initial Hyperrectangle Forming Methods}\label{sec:rect_forming_methods}

Here, we describe some existing methods for forming conformal prediction intervals with conditional mean estimates (more generally, point estimates). One of the most common non-conformity scores for a point estimate is the absolute difference, $V_i = |Y_i - \hat{f}(\vect{X}_i)|$ \citep{papadopoulos_2002, conformal_book}. The prediction interval in the case of the absolute difference the interval is
\[
C(\vect{X}_{n+1}) = [\hat{f}(\vect{X}_{n+1}) - Q_{1-\alpha}(\vect{V}; Z_{cal1}), \hat{f}(\vect{X}_{n+1}) + Q_{1-\alpha}(\vect{V}; Z_{cal1})]
\]
where 
\[
Q_{1-\alpha}(\vect{V}; Z_{cal1}) := (1-\alpha)(1 + \frac{1}{|Z_{cal1}|})-\text{th empirical quantile of }\{V_i\}.
\]
and $|Z_{cal1}|$ is the size of the calibration set. We refer to this method as Mean Absolute method in Section~\ref{sec:simulation}.

A second option that allows one to control the tail error rates is signed-conformal regression. Define the signed error non-conformity scores, $R_i = Y_i - \hat{f}(\vect{X}_i)$ and $V_i = -1 \times R_i$. Then, the signed error conformal prediction region (SECPR) is given by 
\[
C(\vect{X}_{n+1}) = [\hat{f}(\vect{X}_{n+1}) - Q_{1 - \alpha_1}(\vect{V}; Z_{cal1}), \hat{f}(\vect{X}_{n+1}) + Q_{1 - \alpha_2}(\vect{R}; Z_{cal1})]
\]
where 
% \[
% Q_{\delta}(\vect{V}; Z_{cal}) := (\delta)(1 + \frac{1}{|Z_{cal1}|})-\text{th empirical quantile of }\{V_i\}
% \]
% \marginpar{Do I need to define this twice?}
% and
$\alpha = \alpha_1 + \alpha_2.$
For example, if one wanted a conformal prediction interval with equal tailed errors, they could take $\alpha_1 = \alpha_2 = \alpha/2$ \citep{Linusson_2014_signed_conformal, romanocqr}. We refer to this method as Mean Signed in Section~\ref{sec:simulation}. 

\section{Quantile Regression} \label{sec:quant_rectangle}

We now describe an extension of conformal hyperrectangle for regression to non-constant prediction interval widths, conformal quantile hyperrectangular regression. As with other split conformal prediction methods, we begin by splitting our data into sets used for training and calibration. The training set is indexed by $\mathcal{I}_{tr}$ and the calibration set by $\mathcal{I}_{cal}$. Given any quantile regression algorithm, $\mathcal{\vect{A}}$, we fit 2 $\times$ $p$ conditional quantile functions $\hat{\vect{q}}_{\lo}(\vect{x})$ and $\hat{\vect{q}}_{\hi}(\vect{x})$ on the training set. Any quantiles will work as long as the quantile for $\hat{\vect{q}}_{\lo}(\vect{x})$ is smaller than that of $\hat{\vect{q}}_{\hi}(\vect{x})$. In practice, splitting $\alpha$ in two and giving the tails equal probability tends to work well.
\[
\{\hat{\vect{q}}_{\lo}, \hat{\vect{q}}_{\hi}\} \leftarrow \mathcal{\vect{A}}(\{(\vect{Y}_i, \vect{X}_i): i \in \mathcal{I}_{tr}\})
\]

We now form initial intervals for the calibration set based on the trained model. The intervals take the form
\begin{equation}
\vect{C}(\vect{X}_{k}) = [\hat{\vect{q}}_{\lo}(\vect{X}_k), \hat{\vect{q}}_{\hi}(\vect{X}_k)], \text{ } \forall k \in \mathcal{I}_{cal}.
\label{eq:p_dim_quant_rect_formation}
\end{equation}

To keep some balance, we record the length of the $p \times |\mathcal{I}_{cal}|$ sides. Let $A_i$ ($B_i$, $C_i$) denote the length of the first (second, third) interval for the $i$-th data pair in the calibration set. Note that the $A$'s are generally distinct real numbers and so are the $B$'s and $C$'s.  That is, the lengths of the sides for different covariates will generally not be the same, unlike in~\cref{sec:mean_rectangle}. 

Now, using the intervals found in \eqref{eq:p_dim_quant_rect_formation}, we compute a non-conformity score to make sure our rectangular prediction regions will have the coverage described in \eqref{eq:p_dim_response_coverage}. For each of the $p$-dimensions of the response, compute the following non-conformity scores:
\[
E_{k,j} = \max\{\hat{q}_{j, \lo}(\vect{X}_k) - Y_{k,j}, Y_{k,j} - \hat{q}_{j, \hi}(\vect{X}_k)\}, \ \forall k \in \mathcal{I}_{cal} \text{ and } j = 1, \ldots, p.
\]
We now convert $E_{k,j}$ to the length of the first dimension using the side lengths we found earlier. Denote these converted scores as $W_{k,j}$. For clarity, some examples of conversions would be $W_{11} = E_{11}$, $W_{12} = E_{12} \times \frac{A_1}{B_1}$, $W_{21} = E_{21}, W_{22} = E_{22} \times \frac{A_2}{B_2}$, etc. Now, define $W_k = \max\limits_{j = 1, \ldots, p} W_{k,j}$. 

Our final steps are to find the adjustments for each of the dimensions, and construct the final prediction region. First, we form the intervals for the data with unseen responses,
\begin{equation}
\vect{C}(\vect{X}_{n+1}) = [\hat{\vect{q}}_{\lo}(\vect{X}_{n+1}), \hat{\vect{q}}_{\hi}(\vect{X}_{n+1})].
\label{eq:p_quant_unadj}
\end{equation}

Now, define $Adj_1$ to be the  $(1 - \alpha)(1 + \frac{1}{|Z_{cal}|})-\text{th empirical quantile of }\{W_k\}$. This is the final adjustment for the first dimension. We can find the adjustments for the other dimensions by undoing the conversion to the first side. To do this, we'll need to find the lengths of the 1st, 2nd, 3rd, etc. sides in \eqref{eq:p_quant_unadj}. Denote these lengths as $A_{n+1}$, $B_{n+1}$, $C_{n+1}$, etc. Then, the final adjustments for the other sides would be $Adj_2 = Adj_1 \times \frac{B_{n+1}}{A_{n+1}}$, $Adj_3 = Adj_1 \times \frac{C_{n+1}}{A_{n+1}}$, etc. Finally, given new input data, $\vect{X}_{n+1}$, the prediction region for $\vect{Y}_{n+1}$ is constructed as 
\[
\hat{\vect{C}}(\vect{X}_{n+1}) = [\hat{\vect{q}}_{\lo}(\vect{X}_{n+1}) - \vect{Adj}, \hat{\vect{q}}_{\hi}(\vect{X}_{n+1}) + \vect{Adj})].
\]
% \marginpar{Should $\hat{C}$ be $\hat{\vect{C}}$?}
For reference, the procedure is summarized in \cref{alg:p_dim_quant}.

\begin{algorithm}
    \caption{P-dimensional Quantile Conformal Regression}\label{alg:p_dim_quant}
    \textbf{Input:} level $\alpha$, data = $\mathcal{Z} = (\vect{Y}_i, \vect{X}_i)_{i \in \mathcal{I}}$, score function $V$, test point $\vect{x}$, and regression algorithm $\mathcal{\vect{A}}(\vect{X}; \mathcal{D})$ for the conditional quantiles using $\mathcal{D}$ as data \newline
    \textbf{Procedure:}
    \begin{algorithmic}[1]

    \State Split $\mathcal{Z}$ into a training fold $\mathcal{Z}_{tr} \overset{\Delta}{=} (\vect{Y}_i, \vect{X}_i)_{i \in \mathcal{I}_{tr}}$ and a calibration fold $\mathcal{Z}_{cal} \overset{\Delta}{=} (\vect{Y}_i, \vect{X}_i)_{i \in \mathcal{I}_{cal}}$ 
    \State Train the models $\hat{\vect{q}}_{\lo}(\vect{x}; Z_{tr})$ and $\hat{\vect{q}}_{\hi}(\vect{x}; Z_{tr})$ using the training set
    \State Find and record the length of the $p \times |\mathcal{Z}_{cal}|$ sides for the hyperrectangles based on the trained models from step 2. Denote these as $A_k$ for the first side length on the $k$th data set point, $B_k$ for the second side length on the $k$th data set point, etc.
    \State For each $k \in \mathcal{I}_{cal}$ and each $j \in 1, \ldots, p$, compute the score $E_{k,j} = \max\{\hat{q}_{j, \lo}(\vect{X}_{k}) - Y_{k,j}, Y_{k,j} - \hat{q}_{j, \hi} (\vect{X}_{k})\}$
    \State Convert the $E_{k,j}$ to the first dimension using the side lengths found in step 3. $W_{11} = E_{11}$, $W_{12} = E_{12} \times \frac{A_1}{B_1}$, $W_{21} = E_{21}, W_{22} = E_{22} \times \frac{A_2}{B_2}$, etc.
    \State Define $W_k = \max\limits_{j = 1, \ldots, p} W_{k,j}$
    \State Find $Adj_1 = (1 - \alpha)(1 + \frac{1}{|Z_{cal}|})-\text{th empirical quantile of }\{W_k\}$. This is the adjustment for the first dimension
    \State Find the lengths of the $p$ test point sides using the trained models from step 2 for the test point, $\vect{x}$. Denote the lengths of these sides as $A_{n+1}$, $B_{n+1}$, $C_{n+1}$, etc.
    \item The final adjustments for the other sides become $Adj_2 = Adj_1 \times \frac{B_{n+1}}{A_{n+1}}$, $Adj_3 = Adj_1 \times \frac{C_{n+1}}{A_{n+1}}$, etc.
    \end{algorithmic}
    \textbf{Output:} $\hat{\vect{C}}(\vect{x}) = (\hat{\vect{q}}_{\lo}(\vect{x}) - \vect{Adj}, \hat{\vect{q}}_{\hi}(\vect{x}) + \vect{Adj})$, where $\hat{\vect{q}}_{\lo}(\vect{x})$ and $\hat{\vect{q}}_{\hi}(\vect{x})$ are the p-dimensional prediction sets for $\vect{x}$ defined in step 2 and $\vect{Adj}$ is the p-dimensional vector of adjustments found in steps 7 and 9

\end{algorithm}

\subsection{Remarks}
One unanswered question is, what coverage rate should the initial hyperrectangle be formed with for both point and quantile regression? We recommended $1-\alpha$ in both~\cref{alg:p_dim_mean} and~\cref{alg:p_dim_quant}. We recommended $1 - \alpha$ and not $1 - \alpha/ \binom{p}{2}$ because it is easier for a maximum non-conformity score (which ours is) to expand rather than contract. If a smaller volume is preferred, and at least one dimension has a skewed error, one can do a grid search starting at a value less than $1 - \alpha$ and increasing up to $1 - \alpha/\binom{p}{2}$. As a general note, this search can also be done for quantile regression in the one-dimensional case to reduce the interval length.

A note on which dimension to use as the reference. First, we will assume that the ratio is constant, as it is in~\cref{sec:mean_rectangle}. Denote the non-conformity scores for the $k$-th observation from the second calibration set in the $j$-th dimension as $E_{k,j}$. The overall $k$-th non-conformity score with the first dimension as the reference is then,
\[
W^1_k = \max\Big\{E_{k,1}, E_{k, 2}\times \frac{A}{B}, \ldots, E_{k, p} \times \frac{A}{P}\Big\}.
\]
With the second dimension as the reference, the score would then be,
\[
W^2_k = \max\Big\{E_{k,1} \times \frac{B}{A}, E_{k, 2}, \ldots, E_{k, p} \times \frac{B}{P}\Big\}.
\]
With the last dimension as the reference, the score would then be,
\[
W^p_k = \max\Big\{E_{k,1} \times \frac{P}{A}, E_{k, 2} \times \frac{P}{B}, \ldots, E_{k, p}\Big\}.
\]
It is clear that, though $W^j_k$ will be different depending on which of the $j$ dimensions is the reference, the (potentially) scaled $E_{k, j}$ that is selected will be the same. That is, if $E_{k, 2} \times \frac{A}{B}$ was the maximum with the first dimension as the reference, $E_{k, 2}$ will be the maximum with the second dimension as the reference. So, the final score that is chosen when the first dimension is the reference will always be the same as the final score selected (scaled by a constant) when the $j$-th dimension is the reference. Now, let $W_q$ represent the chosen final non-conformity score. If the reference dimension is the first dimension, we have that,
\[
W^1_q = \max \Big\{E_{q, 1}, E_{q, 2} \times \frac{A}{B}, \ldots, E_{q, p} \times \frac{A}{P}\Big\}.
\]
With no loss of generality, assume that $E_{q, 1}$ is the maximum of the two. Then, the final adjustments are,
\[
Adj_1 = E_{q, 1},
\]
\[
Adj_2 = E_{q, 1} \times \frac{B}{A},
\]
and
\[
Adj_p = E_{q, 1} \times \frac{P}{A}.
\]
Next, assume that the second dimension was the reference. We then have that,
\[
W^2_q = \max \Big\{E_{q, 1} \times \frac{B}{A}, E_{q, 2}, \ldots, E_{q, p} \times \frac{B}{P} \Big\}.
\]
Because $E_{q, 1}$ was the maximum before, $E_{q, 1} \times \frac{B}{A}$ will be the maximum for the above and our final adjustments are,
\[
Adj_1 = E_{q, 1} \times \frac{B}{A} \times \frac{A}{B} = E_{q, 1},
\]
\[
Adj_2 = E_{q, 1} \times \frac{B}{A},
\]
and
\[
Adj_p = E_{q, 1} \times \frac{B}{A} \times \frac{P}{B} = E_{q, 1} \times \frac{P}{A}.
\]
Because the second dimension is general, the reference dimension does not change the results when the length ratios are a fixed quantity.

\begin{comment}
\textbf{A note on which side to use as the reference. Assume that the ratio is constant, as it is in~\cref{sec:mean_rectangle}. Then it is clear that with two sides, the reference side does not make a difference. It turns out that the results will be the same regardless of how many dimensions the response is. Imagine that you are using the $k$th dimension as the reference. Then we have that a-adjust = $c_{final} \times A/B \times B/C \times \ldots \times J/K = k_{final} \times A/K$. Because we know there is no difference in two dimensions, there is no difference going from $a$ to $b$. Then from $b$ to $c$, and again until we get to the reference dimension.}
\end{comment}

If our ratios are not the same, as in~\cref{sec:quant_rectangle}, then the reference dimension can make a minor difference. Though, we've found in simulation studies that the difference in volume between different reference dimensions is small. It ultimately depends on the ratio output by the quantile regression models that are being used. Because the reference dimension has a fixed length adjustment, regardless of how much it varies with the covariates, it does make some intuitive sense to set the ``most homoskedastic" dimension to be the reference. One potential approach to finding the most homoskedastic dimension is to look at the variability of the side lengths output in the calibration set. The less variable the side lengths are, the more homoskedastic the dimension should be, assuming a reasonably good model. 

\section{Validity of Conformal Hyperrectangles}\label{sec:validity}

\begin{theorem}\label{thm:rec_validity}
For both conformal hyperrectangular regression and conformal quantile hyperrectangular regression,
\[
\mathbb{P}(\vect{Y}_{n+1} \in \hat{\vect{C}}(\vect{X}_{n+1})) \geq 1 - \alpha
\]
\noindent If the $V_i$'s are almost surely distinct, then
\[
\mathbb{P}(\vect{Y}_{n+1} \in \hat{\vect{C}}(\vect{X}_{n+1})) \leq 1 - \alpha + \frac{1}{n_{cal} + 1}.
\]
\end{theorem}

\noindent The proof is provided in~\cref{prf:rec_validity}.

\section{Balance}\label{sec:balance}
% Comment: I think we should assume that $E(\epsilon) = 0$. Otherwise, it should just be in the mean function. Maybe more generally that $E(\sigma(\vect{X}\epsilon|\vect{X}) = 0$?
Here, we introduce two assumptions under which the proposed methods can achieve balance in some sense. Throughout this section, we assume that $(\vect{X}_i, \vect{Y}_i) \overset{i.i.d.}{ \sim} P$.
\newline 
{\bf Assumption ICED}: There exist $\mu_i, c_i, i=1,\ldots,p$, where the $c$'s are positive constants and the $\mu$'s are functions of $\vect{X}$ such that $c_1\{Y_1 - \mu_1(\vect{X})\}|\vect{X}\sim c_2\{Y_2 - \mu_2(\vect{X})\}|\vect{X}\sim \cdots \sim c_p\{Y_p - \mu_p(\vect{X})\}|\vect{X}$, that is,  the conditional distributions of $c_i\{Y_i - \mu_i(\vect{X})\}$ given $\vect{X}$ are identical for $i=1,\ldots,p$, almost surely.
\newline
{\bf Assumption CQ}: Assume the ICED assumption holds and denote $G$ as the common distribution of $c_i\{Y_i-\mu_i(\vect{X})\}$. Assume that the quantile function of  $G$ is a continuous function.

The first assumption will be simply referred to as the identical conditional error distribution (ICED) assumption, while the second assumption is referred to as the continuous quantile (CQ) assumption. 

We now introduce the concept of asymptotic balance for conformal hyperrectanges.
\begin{definition}
Asymptotic Balance

\noindent The hyperrectangle $\prod_j  \hat{C}_j$ is said to achieve asymptotic balance as $n_{cal}$ and $n_{tr} \to \infty$ if
\[
\max\limits_j \{|\eta_j - \frac{1}{p}\sum\limits_{i = 1}^p \eta_i|\} \to 0 \text{ in probability},
\]
where $\eta_j := P(Y_j \notin \hat{C}_j), \ \forall j = 1, \ldots , p$, where $\hat{C}_j$ is the marginal conformal prediction interval for the $j$-th dimension.

%\noindent Define $\eta_j := P(Y_j \notin \hat{C}_j)$ for some set, $\hat{C}_j$. Then, the hyperrectangle $\prod_j  \hat{C}_j$ is said to achieve asymptotic balance if it holds that  
%\[
%\max\limits_j \{|\eta_j - \frac{1}{p}\sum\limits_{i = 1}^p \eta_i|\} \to 0 \text{ in probability}.
%\]
\end{definition}

We now state the main results whose proofs are given in~\cref{prf:quant_rec_bal}.
\begin{theorem}\label{thm:quant_rec_bal}
    For $\vect{X} = \vect{x} \in \mathcal{C}$, where $\mathcal{C}$ is a compact set, under the ICED assumption when the initial unconformalized hyperrectangle is formed with a consistent estimator for the conditional quantiles (that is, any conditional $\alpha\times 100\%$-quantile estimate with fixed   $\alpha\in (0,1)$ is  uniformly consistent for $\vect{x}\in \mathcal{C}$),  conformal quantile hyperrectangular regression achieves asymptotic balance between dimensions for $\alpha \in (0, 1)$ as $n_{tr} \to \infty$.
\end{theorem}

\begin{theorem}\label{thm:point_rec_bal}
    For $\vect{X} = \vect{x} \in \mathcal{C}$, where $\mathcal{C}$ is a compact set, under the  CQ assumption when the initial unconformalized hyperrectangle is formed with an estimator for $\mu_j$ such that $ \sup_{x \in \mathcal{C}} |\hat{\mu}_j(\vect{x}) - \mu_j(\vect{x})| \overset{a.s.}{\to} 0 \text{ as } n_{tr} \to \infty \ \forall j$ and an initial non-conformity score of $V_{i,j} = |Y_{i,j} - \mu_j(\vect{X})|$, conformal hyperrectangular regression achieves asymptotic balance between dimensions for $\alpha \in (0, 1)$, as  $n_{tr}\to \infty$ and $n_{cal1} \to \infty$. 
\end{theorem}

% The proof is provided in~\cref{prf:point_rec_bal}.

% In practice, we will never have an infinite training sample size for a consistent estimator. 
The preceding results show that, under the ICED assumption or the CQ assumption and with a large training and calibration sample size and a reasonably good algorithm, it is expected that the hyperrectangles' coverage rates are expected to be reasonably balanced. Under the same assumptions, we could use a Bonferroni correction to achieve balance, but our methods achieve balance as well as the desired nominal coverage thanks to~\cref{thm:rec_validity}. A Bonferroni correction would overcover both marginally and overall. In otherwords, a Bonferroni correction is marginally (and overall) conservative, while our methods allow us to find the right marginal coverage rates that lead to tight overall coverage. This is corroborated by the numerical results given in Section~\ref{sec:simulation}. 

Next, we look at a favorable property of balance besides fairness. If instead of the ICED assumption, we make the much stronger assumption that the responses between dimensions are conditionally independent and identically distributed, that is, $Y_j |\vect{X}$ are i.i.d.  

%, specifically,  there exists some latent random variable, $\xi$, such that $Y_1|\vect{X}, \xi, \ldots, Y_p|\vect{X}, \xi$ are i.i.d. 
% If we want to bring the above back, I saved that version in overleaf

Consider a hyperrectangle of the form $\prod_{i=1}^p R_i$. Suppose the average marginal conditional coverage rate of the the marginal prediction sets $R_1, \ldots, R_p$ is equal to a fixed constant, say, $\tau>0$, that is, $\sum_{j = 1}^p P(Y_j \in R_j|\vect{X})\equiv p\tau$. We claim that the maximum joint coverage probability of  such hyperrectangles equals $\tau^p$, which is attained if the hyperrectangle is furthermore balanced, that is,  $\forall j, \ P(Y_j \in R_j|\vect{X})\equiv \tau$. To see this, by independence we have that,
\[
1 - \alpha = P(Y_1 \in R_1, \ldots, Y_p \in R_p|\vect{X})
\]
\[
=\prod_{j = 1}^p P(Y_j \in R_j|\vect{X}).
\]
Because the geometric mean is not greater  than the arithmetic mean, it is furthermore
\begin{equation*}\label{eq:balance_idea}
\leq  \Bigg[\sum_{j = 1}^p \frac{P(Y_j \in R_j|\vect{X})}{p} \Bigg]^p,
\end{equation*}
with equality if and only if $P(Y_j \in R_j|\vect{X})$ are the same $\forall j$, which must then be $\tau$. This completes the proof of the claim. 
% Again, by the law of total expectation,
% \[
% = \Big[\sum_{j = 1}^p \frac{P(Y_j \in R_j|\vect{X})}{p} \Big]^p.
% \]
% The sum is clearly minimized when the inequality in \eqref{eq:balance_idea} is an equality, so the sume is minimized when we have balance. It is also equivalent to minimizing,
% \[
% \sum_{j = 1}^p P(Y_j \in R_j|\vect{X}).
% \]
% Thus, when we have balance, the above is minimized. 
By achieving balance under the independence between dimensions assumption, we are minimizing the sum of marginal coverage rates while maintaining the overall coverage of the prediction region. This allows our prediction intervals to be sharper than other methods with the same overall coverage. 

The inequality that the product of coverage probabilities is less than the joint probability was proven in \cite{sidak_1967_square_normal} for Normal random variables. \cite{Dykstra_cov_bound1973} extended it to responses that are conditionally independent and identically distributed between dimensions, the conditioning variable need not be observed, so the case where the response variables are (conditionally) exchangeable between dimensions applies to the bound:
\begin{equation}\label{eq:joint_leq_prod}
P(Y_1 \in R_1, \ldots, Y_p \in R_p|\vect{X}) \geq \prod_{j = 1}^p P(Y_j \in R_j|\vect{X}).
\end{equation}
Indeed, even when the independence assumption is broken, our simulation results show that when the responses are exchangeable between dimensions, marginal coverage rates are just slightly below $(1 - \alpha) ^ {1/p}$. 
%The inequality that the product of coverage probabilities is less than the joint probability was proven in \textcite{sidak_1967_square_normal} for Normal random variables. \textcite{Dykstra_cov_bound1973} extended it to responses that are conditionally independent and identically distributed between dimensions for the same set $R$. The conditioning variable need not be observed, so the case where the response variables are exchangeable between dimensions applies. We extend to the bound:
%\[
%P(Y_1 \in R_1, \ldots, Y_p \in R_p|\vect{X}) \leq \prod_{j = 1}^p P(Y_j \in R_j|\vect{X}),
%\]
%to the case where we have multiple balanced sets, $R_j$, in a similar fashion to \textcite{Dykstra_cov_bound1973}.
%\begin{align}
%P(Y_1 \in R_1, \ldots, Y_p \in R_p|\vect{X}) &= \int \Bigg(\int_{R_1} dF \ldots \int_{R_1} dF\Bigg) d\mu(F) \\
%&= E\Bigg[\Bigg(\int_{R_1} dF\Bigg) ^ p\Bigg] \geq \Bigg[E\Bigg(\int_{R_1} dF\Bigg)\Bigg]^p\\
%&= \prod_{j = 1}^p P(Y_j \in R_j|\vect{X})
%\end{align}

As a note of interest, \cite{Dykstra_cov_bound1973} first observed that as the probabilities become large, the product of the marginal coverage probabilities approaches the joint coverage probability when the marginal probabilities are pairwise negatively dependent. Formally that is,
\begin{equation*}\label{eq:log_bound_marg_prod}
|P(Y_1 \in R_1, \ldots, Y_p \in R_p|\vect{X}) - \prod_{j = 1}^p P(Y_j \in R_j|\vect{X})| \leq \frac{p - 1}{2p} \Big[ \log \Big\{\prod_{j = 1}^p P(Y_j \in R_j|\vect{X})\Big\}\Big]^2.
\end{equation*}
Marginal coverage probabilities that make up a joint prediction set are not necessarily pairwise negatively dependent, though, they can be in the exchangeable response case.

\section{Simulation Studies}\label{sec:simulation}
\subsection{Comparison of Four Methods}
We looked at three dimensional responses with a varying range of correlation between the error terms. In each simulation, the model was correctly specified. We compared five approaches to forming the conformal prediction sets. The first two used the conditional mean approach found in~\cref{sec:mean_rectangle}. Data were split as follows: 500 for training, 250 for the first calibration, and 250 for the second calibration. To form the initial hyperrectangles, we used the absolute difference non-conformity score (Mean Absolute) and the signed error non-conformity score (Mean Signed) with equal tails. 
The third approach was our proposed quantile regression approach found in~\cref{sec:quant_rectangle} (Quantile), with equal tails for the initial quantile model. Data were split; 500 for training and 500 for calibration. The final was the maximum absolute difference, $V = \max\limits_{j = 1, \ldots, q}|Y_j - \hat{f}_j(\vect{X})|$ (Absolute Max), where $\hat{f}_j(\vect{X})$ was an estimate of the conditional mean. Data were split; 500 for training and 500 for calibration. All results were rounded to three decimal places, with simulation standard errors in parentheses if they were larger than 0.0005. 

%https://math.stackexchange.com/questions/163470/generating-correlated-random-numbers-why-does-cholesky-decomposition-work

The data were generated as follows, $X_1 \sim \text{exp}(\lambda = 0.2)$, $X_2 \sim \text{Unif}(-5, 5)$, $E(Y_1|\vect{X}) = 5 + 2X_1$, $E(Y_2|\vect{X}) = 3X_1 + X_1X_2$, and $E(Y_3|\vect{X}) = X_2^2$. The errors for the responses were generated as $\epsilon_1 \sim \text{Gamma}(\alpha = 2, \lambda = 0.2)$, $\epsilon_2 \sim \text{Gamma}(\alpha = 3, \lambda = 0.5)$, and $\epsilon_3 \sim \mathcal{N}(5X_2, 1)$ where $\lambda$ represents the rate parameter and $\alpha$ represents the shape parameter. We started off with the uncorrelated errors for all of the responses that are listed above, $\vect{\epsilon} = (\epsilon_1, \epsilon_2, \epsilon_3)^T$, and a desired correlation matrix between them, $R$. To generate the correlated errors, we multiplied $\vect{\epsilon}$ by the upper-triangular Cholesky decomposition of $R$. For each simulation, the mean or quantile structure of the model was correctly specified. All parameters in the specified model were estimated. The miscoverage rate, $\alpha$, was set to be 0.10 in all simulations. The simulation size was $1,000,000$ in all simulations. In the four simulations the error correlation matrix was
\[
R_1 = \begin{bmatrix}
1 & 0.3 & 0.6\\
0.3 & 1 & 0.5 \\
0.6 & 0.5 & 1
\end{bmatrix}\,\,
% \]
% \[
R_2 = \begin{bmatrix}
1 & 0.8 & 0.8\\
0.8 & 1 & 0.8 \\
0.8 & 0.8 & 1
\end{bmatrix}\,\,
% \]
% \[
R_3 = \begin{bmatrix}
1 & 0.6 & 0.3\\
0.6 & 1 & 0.5 \\
0.3 & 0.5 & 1
\end{bmatrix}\,\,
% \]
% \[
R_4 = \begin{bmatrix}
1 & 0.2 & 0.2\\
0.2 & 1 & 0.2 \\
0.2 & 0.2 & 1
\end{bmatrix}.
\]

The results of the simulations are given in Tables \ref{tab:p_dim_sim1} - \ref{tab:p_dim_sim4}. Coverage is the average coverage over all simulations, volume is the average volume. Len i and Marg i represent the length of the interval in the i-th dimension and the marginal coverage in the i-th dimension, respectively. 

\begin{landscape}
\begin{table}
\begin{center}
    \begin{tabular}{|c|c|c|c|c|c|c|c|c|}
        \hline
         Approach &  Coverage & Volume & Len 1 & Len 2 & Len3 & Marg 1 & Marg 2 & Marg 3\\
         \hline
         Mean Absolute  & 0.899 & 23,499.007 (4.656) &25.144 & 15.360 & 60.307 & 0.939 & 0.954 & 0.983\\
         Mean Signed  & 0.900 & 21,008.494 (2.377) & 25.711 & 14.776 & 55.242 &0.959 & 0.950 & 0.968\\
        %HPD & 0.900 & 24,549.708 (5.759) &25.156 & 15.675 & 61.380 & 0.939 & 0.952 & 0.984\\

        Quantile  & 0.900 & 15,274.272 (2.224) & 25.199 & 14.409 & 42.165 &0.957 & 0.943 & 0.968\\
         Absolute Max & 0.900 & 110,634.273 (8.948) &47.972 & 47.972 & 47.972 & 0.991 & 1.000 & 0.906\\

         \hline
    \end{tabular}
    %\captionsetup{justification=centering}
    \caption{Simulation Result: Setup Number 1}
    \label{tab:p_dim_sim1}
\end{center}
\end{table}
% \end{landscape}

% \begin{landscape}
\begin{table}
\begin{center}
    \begin{tabular}{|c|c|c|c|c|c|c|c|c|}
        \hline
         Approach &  Coverage & Volume & Len 1 & Len 2 & Len 3 & Marg 1 & Marg 2 & Marg 3\\
         \hline
         Mean Absolute  & 0.899 & 28,927.063 (5.099) & 23.457 & 21.044 & 58.178 & 0.930 & 0.937 & 0.965\\
         Mean Signed  & 0.900 & 29,657.267 (3.652) & 24.841 & 21.556 & 55.271 &0.953 & 0.944 & 0.952\\
         %HPD  & 0.900 & 29,378.517 (5.811) &23.287 & 21.289 & 58.667 & 0.930 & 0.935 & 0.966\\

        Quantile  & 0.900 & 22,086.655 (3.165) & 24.120 & 20.885 & 43.764 &0.946 & 0.934 & 0.951\\
         Absolute Max & 0.900 & 121,806.996 (10.909) &49.527 & 49.527 & 49.527 & 0.992 & 0.997 & 0.904\\

         \hline
    \end{tabular}
    %\captionsetup{justification=centering}
    \caption{Simulation Result: Setup Number 2}
    \label{tab:p_dim_sim2}
\end{center}
\end{table}
\end{landscape}

\begin{landscape}
\begin{table}
\begin{center}
    \begin{tabular}{|c|c|c|c|c|c|c|c|c|}
        \hline
         Approach &  Coverage & Volume & Len 1 & Len 2 & Len 3 & Marg 1 & Marg 2 & Marg 3\\
         \hline
         Mean Absolute  & 0.899 & 24,423.676 (4.319) &23.804 & 18.501 & 55.077 & 0.932 & 0.943 & 0.980\\
         Mean Signed  & 0.899 & 23,606.321 (2.528) &24.883 & 18.357 & 51.620 & 0.954 & 0.942 & 0.968\\
         %HPD & 0.900 & 25,050.237 (5.247) &23.664 & 18.763 & 55.763 & 0.932 & 0.943 & 0.982\\

        Quantile  & 0.900 & 16,050.704 (2.470) &24.372 & 17.941 & 36.675 & 0.949 & 0.935 & 0.969\\
         Absolute Max & 0.900 & 100,115.258 (7.246) &46.408 & 46.408 & 46.408 & 0.990 & 0.999 & 0.908\\

         \hline
    \end{tabular}
    %\captionsetup{justification=centering}
    \caption{Simulation Result: Setup Number 3}
    \label{tab:p_dim_sim3}
\end{center}
\end{table}
% \end{landscape}

% \begin{landscape}
\begin{table}
\begin{center}
    \begin{tabular}{|c|c|c|c|c|c|c|c|c|}
        \hline
         Approach &  Coverage & Volume & Len 1 & Len 2 & Len 3 & Marg 1 & Marg 2 & Marg 3\\
         \hline
         Mean Absolute  & 0.899 & 21,768.531 (4.820) &25.522 & 14.445 & 58.368 & 0.941 & 0.954 & 0.993\\
         Mean Signed  & 0.899 & 18,133.270 (2.036) &25.441 & 13.365 & 52.243 & 0.958 & 0.948 & 0.983\\
         %HPD & 0.900 & 23,143.081 (5.931) &25.671 & 14.823 & 59.772 & 0.942 & 0.953 & 0.995\\

        Quantile  & 0.900 & 12,000.470 (1.938) &25.068 & 13.410 & 35.678 & 0.956 & 0.943 & 0.985\\
         Absolute Max & 0.900 & 96,679.263 (6.951) &45.874 & 45.874 & 45.874 & 0.990 & 1.000 & 0.909\\

         \hline
    \end{tabular}
    %\captionsetup{justification=centering}
    \caption{Simulation Result: Setup Number 4}
    \label{tab:p_dim_sim4}
\end{center}
\end{table}
\end{landscape}

We can see from the simulation results that all four methods tended to achieve the desired $90\%$ overall coverage. The absolute maximum approach has poor length and volume compared to the other methods, and the balance (spreading out the miscoverage between dimensions) is poor when the variance of the error term for at least one dimension is much larger than other dimensions. The signed and quantile approaches tend to have the best balance, but because of the heteroskedastic error for the third response, the quantile rectangular approach does a significantly better job with respect to the volume and the length of the intervals.  

\subsection{Ten-dimensional Response}
A second simulation was done for a ten-dimensional response. The simulation setup is given in~\cref{tab:10_dim_setup1} and~\cref{tab:10_dim_setup2}. The results are given in~\cref{tab:10_dim_results}. Only the quantile method of forming the rectangular prediction region was used. Data were split; 500 for training and 500 for calibration. All of the $\epsilon$ terms were generated from a multivariate Normal distribution with a mean of 0, a variance of 1, and a correlation of 0.5. That is, $\epsilon_i \sim \mathcal{N}(0, 1) \text{ } i = 1, \ldots, 10$ and $corr(\epsilon_i, \epsilon_j) = 0.5$ for $i \neq j$ and $i, j = 1, \ldots, 10$. All quantile models were fit with the correctly specified conditional quantiles, which can be found in~\cref{tab:10_dim_setup2}. We can also see that some of the error terms had a heteroskedastic term added to them. The miscoverage rate, $\alpha$, was set to be 0.10. The simulation size was $1,000,000$. The results of the simulation can be found in~\cref{tab:10_dim_results}.

\begin{table}
\begin{center}
    \begin{tabular}{|c|c|c|c|}
        \hline
         Variable &  Distribution & Parameters \\
         \hline
         $X_1$ & Uniform & min = -2; max = 5\\
         $X_2$ & Uniform & min = -5; max = -1\\
        $X_3$ & Uniform & min = -6; max = 10\\
         $X_4$ & Uniform & min = 0; max = 4\\
         $X_5$ & Uniform & min = $X_2$; max = $X_4$\\
         \hline
    \end{tabular}
    %\captionsetup{justification=centering}
    \caption{Simulation Setup: Ten Dimensions}
    \label{tab:10_dim_setup1}
\end{center}
\end{table}

\begin{table}
\begin{center}
    \begin{tabular}{|c|c|c|c|}
        \hline
         Variable & Conditional Mean & Error \\
         \hline
         $Y_1$ & $2X_1$ & $\epsilon_1$ \\
         $Y_2$ & $X_1 + X_1 X_2$ & $\epsilon_2$\\
        $Y_3$ & $X_2^2$ & $\epsilon_3 + 5X_2$ \\
         $Y_4$ & $X_2 X_5$ & $\epsilon_4 + X_5^2$\\
         $Y_5$ & $X_5^2$ & $\epsilon_5$\\
         $Y_6$ & $X_1^2$ & $\epsilon_6 - 2X_1$\\
         $Y_7$ & $X_4^2$ & $\epsilon_7 - X_4$\\
         $Y_8$ & $X_3^2$ & $\epsilon_8$ \\
         $Y_9$ & $X_4^2$ & $\epsilon_9$ \\
         $Y_{10}$ & $X_1X_2$ & $\epsilon_{10}$\\
         \hline
    \end{tabular}
    %\captionsetup{justification=centering}
    \caption{Simulation Setup: Ten Dimensions}
    \label{tab:10_dim_setup2}
\end{center}
\end{table}

\begin{table}
\begin{center}
    \begin{tabular}{|c|c|c|c|}
        \hline
         Dimension &  Coverage & Length \\
         \hline
         Overall  & 0.900 & NA \\
         First & 0.987 & 5.043  \\
         Second & 0.986 & 5.040  \\
         Third & 0.990 & 7.048  \\
         Fourth & 0.981 & 17.506 (0.001)  \\
         Fifth & 0.987 & 5.045  \\
         Sixth & 0.995 & 11.457 (0.001)  \\
         Seventh & 0.987 & 5.251  \\
         Eighth & 0.987 &  5.044 \\
         Ninth & 0.987 & 5.045  \\
         Tenth & 0.987 & 5.045 \\

         \hline
    \end{tabular}
    %\captionsetup{justification=centering}
    \caption{Simulation Result: Ten Dimensions}
    \label{tab:10_dim_results}
\end{center}
\end{table}

We can see, when the error terms are similar (in this case, the same for dimensions 1, 2, 8, 9, 10), the balance is nearly perfect. Marginal coverage tends to be higher for the dimensions with a heteroskedastic error, though the length is not unreasonable. It is also clear that the overall desired coverage of $90\%$ is achieved. 

\subsection{Checking Balance Between Dimensions}

Two simulation scenarios were run with a 10 dimensional response. Only the quantile method of forming the rectangular prediction region was used, one with a homoskedastic error and one with a heteroskedastic error. Data were split; 500 for training and 500 for calibration. In the homoskedastic scenario, all of the $\epsilon$ terms were generated from a multivariate Normal distribution with a mean of 0, a variance of 1, and a correlation of 0.9. That is, $\epsilon_i \sim \mathcal{N}(0, 1) \text{ } i = 1, \ldots, 10$ and $corr(\epsilon_i, \epsilon_j) = 0.9$ for $i \neq j$ and $i, j = 1, \ldots, 10$. In the heteroskedastic scenario, all of the $\epsilon|\vect{X}$ terms were generated from a multivariate Normal distribution with a mean of 0, a variance of $|X_1|$, and a correlation of 0.9. That is, $\epsilon_i|\vect{X} \sim \mathcal{N}(0, |X_1|) \text{ } i = 1, \ldots, 10$ and $corr(\epsilon_i|\vect{X}, \epsilon_j|\vect{X}) = 0.9$ for $i \neq j$ and $i, j = 1, \ldots, 10$. All quantile models were fit with the correctly specified conditional quantiles, which can be found in~\cref{tab:10_dim_setup_homoskedastic} and~\cref{tab:10_dim_setup_heteroskedastic}. Covariates were generated according to~\cref{tab:10_dim_setup1}. The results of the simulation scenario for the homoskedastic setup can be found in~\cref{tab:10_dim_results_homoskedastic}. For the heteroskedastic scenario, they can be found in~\cref{tab:10_dim_results_heteroskedastic}. The simulation size was 1,000,000. We also compared our conformal quantile hyperrectangular regression with marginal conformalized quantile regression and a Bonferroni correction. Those results can be found in~\cref{tab:10_dim_results_homoskedastic_bon} and~\cref{tab:10_dim_results_heteroskedastic_bon}.

\begin{table}
\begin{center}
    \begin{tabular}{|c|c|c|c|}
        \hline
         Variable & Conditional Mean & Error \\
         \hline
         $Y_1$ & $2X_1$ & $\epsilon_1$ \\
         $Y_2$ & $X_1 + X_1 X_2$ & $2\epsilon_2$\\
        $Y_3$ & $X_2^2$ & $3\epsilon_3$ \\
         $Y_4$ & $X_2 X_5$ & $4\epsilon_4$\\
         $Y_5$ & $X_5^2$ & $5\epsilon_5$\\
         $Y_6$ & $X_1^2$ & $\epsilon_6 $\\
         $Y_7$ & $X_4^2$ & $2\epsilon_7 $\\
         $Y_8$ & $X_3^2$ & $3\epsilon_8$ \\
         $Y_9$ & $X_4^2$ & $4\epsilon_9$ \\
         $Y_{10}$ & $X_1X_2$ & $5\epsilon_{10}$\\
         \hline
    \end{tabular}
    %\captionsetup{justification=centering}
    \caption{Simulation Setup: Homoskedastic Ten Dimensions}
    \label{tab:10_dim_setup_homoskedastic}
\end{center}
\end{table}

\begin{table}
\begin{center}
    \begin{tabular}{|c|c|c|c|}
        \hline
         Variable & Conditional Mean & Error \\
         \hline
         $Y_1$ & $2X_1$ & $\epsilon_1$ \\
         $Y_2$ & $X_1 + X_1 X_2$ & $2\epsilon_2$\\
        $Y_3$ & $X_2^2$ & $3\epsilon_3$ \\
         $Y_4$ & $X_2 X_5$ & $4\epsilon_4$\\
         $Y_5$ & $X_5^2$ & $5\epsilon_5$\\
         $Y_6$ & $X_1^2$ & $\epsilon_6 $\\
         $Y_7$ & $X_4^2$ & $2\epsilon_7 $\\
         $Y_8$ & $X_3^2$ & $3\epsilon_8$ \\
         $Y_9$ & $X_4^2$ & $4\epsilon_9$ \\
         $Y_{10}$ & $X_1X_2$ & $5\epsilon_{10}$\\
         \hline 
    \end{tabular}
    %\captionsetup{justification=centering}
    \caption{Simulation Setup: Heteroskedastic Ten Dimensions}
    \label{tab:10_dim_setup_heteroskedastic}
\end{center}
\end{table}

\begin{table}
\begin{center}
    \begin{tabular}{|c|c|c|c|}
        \hline
         Dimension &  Coverage & Length \\
         \hline
         Overall  & 0.900 & NA \\
         First & 0.965 & 4.257  \\
         Second & 0.964 & 8.508  \\
         Third & 0.965 & 12.777 (0.001)  \\
         Fourth & 0.965 & 17.037  (0.001)\\
         Fifth & 0.965 & 21.295 (0.001) \\
         Sixth & 0.965 & 4.258  \\
         Seventh & 0.965 & 8.517  \\
         Eighth & 0.965 &  12.776 (0.001)\\
         Ninth & 0.965 & 17.036 (0.001)\\
         Tenth & 0.965 & 21.285 (0.001) \\

         \hline
    \end{tabular}
    %\captionsetup{justification=centering}
    \caption{Hyperrectangle Balance Simulation Result: Homoskedastic Ten Dimensions}
    \label{tab:10_dim_results_homoskedastic}
\end{center}
\end{table}

\begin{table}
\begin{center}
    \begin{tabular}{|c|c|c|c|}
        \hline
         Dimension &  Coverage & Length \\
         \hline
         Overall  & 0.901 & NA \\
         First & 0.958 & 6.396  \\
         Second & 0.960 & 12.703 (0.001) \\
         Third & 0.960 & 19.054 (0.001)\\
         Fourth & 0.960 & 25.394 (0.001)\\
         Fifth & 0.960 & 31.747 (0.002) \\
         Sixth & 0.960 & 6.350 \\
         Seventh & 0.960 & 12.699 (0.001) \\
         Eighth & 0.960 &  19.047 (0.001) \\
         Ninth & 0.960 & 25.395 (0.001) \\
         Tenth & 0.960 & 31.738 (0.002)\\

         \hline
    \end{tabular}
    %\captionsetup{justification=centering}
    \caption{Hyperrectangle Balance Simulation Result: Heteroskedastic Ten Dimensions}
    \label{tab:10_dim_results_heteroskedastic}
\end{center}
\end{table}

\begin{table}
\begin{center}
    \begin{tabular}{|c|c|c|c|}
        \hline
         Dimension &  Coverage & Length \\
         \hline
         Overall  & 0.963 & NA \\
         First & 0.990 & 5.381  \\
         Second & 0.990 & 10.979 (0.001)  \\
         Third & 0.990 & 16.152 (0.001)  \\
         Fourth & 0.990 & 21.552  (0.001)\\
         Fifth & 0.990 & 26.945 (0.002) \\
         Sixth & 0.990 & 5.384  \\
         Seventh & 0.990 & 10.768 (0.001)  \\
         Eighth & 0.990 &  16.155 (0.001)\\
         Ninth & 0.990 & 21.529 (0.001)\\
         Tenth & 0.990 & 26.933 (0.002) \\

         \hline
    \end{tabular}
    %\captionsetup{justification=centering}
    \caption{Bonferroni Balance Simulation Result: Homoskedastic Ten Dimensions}
    \label{tab:10_dim_results_homoskedastic_bon}
\end{center}
\end{table}

\begin{table}
\begin{center}
    \begin{tabular}{|c|c|c|c|}
        \hline
         Dimension &  Coverage & Length \\
         \hline
         Overall  & 0.967 & NA \\
         First & 0.990 & 9.346 (0.001)  \\
         Second & 0.990 & 18.357 (0.002) \\
         Third & 0.990 & 26.997 (0.002)\\
         Fourth & 0.990 & 36.033 (0.003)\\
         Fifth & 0.990 & 45.052 (0.004) \\
         Sixth & 0.990 & 8.995 (0.001)\\
         Seventh & 0.990 & 17.996 (0.002) \\
         Eighth & 0.990 &  26.999 (0.002) \\
         Ninth & 0.990 & 35.996 (0.003) \\
         Tenth & 0.990 & 44.993 (0.004)\\

         \hline
    \end{tabular}
    %\captionsetup{justification=centering}
    \caption{Bonferroni Balance Simulation Result: Heteroskedastic Ten Dimensions}
    \label{tab:10_dim_results_heteroskedastic_bon}
\end{center}
\end{table}

These simulation results provide corroboration with~\cref{sec:balance}, that if the conditional error terms follow the same distribution, up to a scale change, our method will achieve good balance and the desired overall coverage with a consistent model.

A second set of simulations were conducted to check the balance when the responses are not conditionally exchangeable between $p$-dimensions, but still follow the ICED assumption. Only the quantile method of forming the hyperrectangular prediction region was used. As in the other simulations, the models were fit with the correctly specified conditional quantiles. The data were generated as follows, $X_1 \sim \text{exp}(\lambda = 0.2)$, $X_2 \sim \text{Unif}(-5, 5)$, $E(Y_1|\vect{X}) = 5 + 2X_1$, $E(Y_2|\vect{X}) = 3X_1 + X_1X_2$, and $E(Y_3|\vect{X}) = 5X_2 + X_2^2$. In the homoskedastic scenario, all of the error terms were generated from a standard Normal distribution. The correlation between the first and second dimension was -0.8. All other correlations between dimensions were 0.8. In the heteroskedastic scenario, all of the error terms were generated from a Normal distribution with a mean of 0 mean and a variance of $|X_1|$. The correlation between the first and second dimension was -0.8. All other correlations between dimensions were 0.8. The simulation sizes were $1,000,000$. The results of the simulation are given in~\cref{tab:3_dim_bal_normal1}.

We looked at a third set of simulations to see if balance is close to holding when the ICED assumption is slightly violated. 
% In each of these simulations covariates and conditional means were generated as follows, 
In each simulation, 
$X_1 \sim \text{exp}(\lambda = 0.2)$, $X_2 \sim \text{Unif}(-5, 5)$, $E(Y_1|\vect{X}) = X_1$, $E(Y_2|\vect{X}) = X_1$, and $E(Y_3|\vect{X}) = X_1$. In the homoskedastic case, the errors were generated as $\epsilon_1|\vect{X} \sim \text{Gamma}(\alpha = 2, \lambda = 0.2)$, $\epsilon_2|\vect{X} \sim \text{Gamma}(\alpha = 2, \lambda = 0.2)$, and $\epsilon_3|\vect{X} \sim \text{Gamma}(\alpha = 2, \lambda = 0.2)$ where $\lambda$ represents the rate parameter and $\alpha$ represents the shape parameter. In the heteroskedastic case, the errors were $\epsilon_1|\vect{X} \sim \text{Gamma}(\alpha = 2|X_2|, \lambda = 0.2)$, $\epsilon_2|\vect{X} \sim \text{Gamma}(\alpha = 2|X_2|, \lambda = 0.2)$, and $\epsilon_3|\vect{X} \sim \text{Gamma}(\alpha = 2|X_2|, \lambda = 0.2)$. We started off with the uncorrelated errors for all of the responses that are listed above, $\vect{\epsilon} = (\epsilon_1, \epsilon_2, \epsilon_3)^T$, and a desired correlation matrix between them, $R$. To generate the correlated errors, we multiplied $\vect{\epsilon}$ by the upper-triangular Cholesky decomposition of $R$. Because linear combinations of gamma distributions do not generally follow a gamma distribution, the ICED assumption is violated. For each simulation, the correctly specified quantile structure of the model was estimated. Data were split; 2,000 for training and 200 for calibration. The miscoverage rate, $\alpha$, was set to be 0.10 in all simulations. The simulation size was $1,000,000$ in all simulations. In the four simulations the error correlation maxtrix was
\[
R_1 = \begin{bmatrix}
1 & 0.3 & 0.6\\
0.3 & 1 & 0.5 \\
0.6 & 0.5 & 1
\end{bmatrix}
\,\,
R_2 = \begin{bmatrix}
1 & 0.8 & 0.8\\
0.8 & 1 & 0.8 \\
0.8 & 0.8 & 1
\end{bmatrix}
\,\,
R_3 = \begin{bmatrix}
1 & 0.6 & 0.3\\
0.6 & 1 & 0.5 \\
0.3 & 0.5 & 1
\end{bmatrix}
\,\,
R_4 = \begin{bmatrix}
1 & 0.2 & 0.2\\
0.2 & 1 & 0.2 \\
0.2 & 0.2 & 1
\end{bmatrix}.
\]
The results of the simulations are given in Tables \ref{tab:3_dim_bal1} - \ref{tab:3_dim_bal4}. Coverage is the average coverage over all simulations, volume is the average volume. Len i and Marg i represent the length of the interval in the i-th dimension and the marginal coverage in the i-th dimension, respectively. 

\begin{landscape}
\begin{table}
\begin{center}
    \begin{tabular}{|c|c|c|c|c|c|c|c|}
        \hline
         Error &  Coverage & Len 1 & Len 2 & Len 3 & Marg 1 & Marg 2 & Marg 3\\
         \hline
         Homoskedastic  & 0.900  &3.981 & 3.977 & 3.978 & 0.951 & 0.950 & 0.950\\
         Heteroskedastic  & 0.900  &7.970 & 7.961 & 7.968 & 0.951 & 0.949 & 0.949\\

         \hline
    \end{tabular}
    %\captionsetup{justification=centering}
    \caption{Simulation Result: Normal Error Three Dimensions}
    \label{tab:3_dim_bal_normal1}
\end{center}
\end{table}
% \end{landscape}

% \begin{landscape}
\begin{table}
\begin{center}
    \begin{tabular}{|c|c|c|c|c|c|c|c|}
        \hline
         Error &  Coverage & Len 1 & Len 2 & Len 3 & Marg 1 & Marg 2 & Marg 3\\
         \hline
         Homoskedastic  & 0.896  &26.155 & 26.499 & 27.073 & 0.963 & 0.959 & 0.954\\
         Heteroskedastic  & 0.897  &41.718 & 41.949 & 42.403 & 0.960 & 0.959 & 0.958\\

         \hline
    \end{tabular}
    %\captionsetup{justification=centering}
    \caption{Simulation Result: Three Dimension Balance Comparison 1}
    \label{tab:3_dim_bal1}
\end{center}
\end{table}
% \end{landscape}

% \begin{landscape}
\begin{table}
\begin{center}
    \begin{tabular}{|c|c|c|c|c|c|c|c|}
        \hline
         Error &  Coverage & Len 1 & Len 2 & Len 3 & Marg 1 & Marg 2 & Marg 3\\
         \hline
         Homoskedastic  & 0.897  &24.733 & 25.413 & 25.504 & 0.953 & 0.944 & 0.941\\
         Heteroskedastic  & 0.897  &42.063 & 42.326 & 42.765 & 0.961 & 0.959 & 0.958\\

         \hline
    \end{tabular}
    %\captionsetup{justification=centering}
    \caption{Simulation Result: Three Dimension Balance Comparison 2}
    \label{tab:3_dim_bal2}
\end{center}
\end{table}
\end{landscape}

\begin{landscape}
\begin{table}
\begin{center}
    \begin{tabular}{|c|c|c|c|c|c|c|c|}
        \hline
         Error &  Coverage & Len 1 & Len 2 & Len 3 & Marg 1 & Marg 2 & Marg 3\\
         \hline
         Homoskedastic  & 0.896  &26.195 & 26.913 & 26.904 & 0.963 & 0.958 & 0.955\\
         Heteroskedastic  & 0.897  &41.751 & 42.286 & 42.267 & 0.960 & 0.958 & 0.958\\

         \hline
    \end{tabular}
    %\captionsetup{justification=centering}
    \caption{Simulation Result: Three Dimension Balance Comparison 3}
    \label{tab:3_dim_bal3}
\end{center}
\end{table}
% \end{landscape}

% \begin{landscape}
\begin{table}
\begin{center}
    \begin{tabular}{|c|c|c|c|c|c|c|c|}
        \hline
         Error &  Coverage & Len 1 & Len 2 & Len 3 & Marg 1 & Marg 2 & Marg 3\\
         \hline
         Homoskedastic  & 0.897  &26.771 & 26.966 & 27.081 & 0.965 & 0.963 & 0.961\\
         Heteroskedastic  & 0.896  &42.661 & 42.794 & 42.879 & 0.964 & 0.963 & 0.962\\

         \hline
    \end{tabular}
    %\captionsetup{justification=centering}
    \caption{Simulation Result: Three Dimension Balance Comparison 4}
    \label{tab:3_dim_bal4}
\end{center}
\end{table}
\end{landscape}

It is clear from these simulation results that balance is nearly achieved in each simulation scenario. A note from these studies is that it appears the correlation between the conditional error terms may impact how good the balance is in finite samples. In~\cref{tab:3_dim_bal4} the marginal coverage rates are much closer to each other, compared to~\cref{tab:3_dim_bal2}, which are the setups with the lowest correlation and highest correlation between dimensions, respectively. We can also see that, although close, the bound in~\eqref{eq:joint_leq_prod} holds as a strict inequality for both the ten dimensional and three dimensional simulations. Balance also appears to nearly hold with minor violations to the ICED assumption.

One point that should be noted from all of the simulations is that the rectangular approach to conformal prediction intervals maintains reasonable balance without sacrificing the volume of the prediction regions. As we can see from the simulations, conformalized hyperrectangles do a much better job than marginal conformal prediction intervals with a Bonferroni correction. The Bonferroni correction suffers from severe overcoverage when the marginal distributions have correlated errors, which is likely to occur with a multivariate response. 

\section{Real Data Analysis}\label{sec:real}

All analyses were performed using R Statistical Software \citep{R_citation}. We looked at how well our conformal quantile hyperrectangle method worked in predicting systolic and diastolic blood pressure for children. The blood pressure data was obtained on children aged 9 and 11 from Bristol, UK via the bp R package \citep{bp_package}. Covariates included in the analysis were the child's gender, age, height, weight, and BMI. The responses of interest were the average systolic and diastolic blood pressure during the subject's first visit. 

The number of subjects in the data set was 1,289. The data were randomly permuted 200 times. For each permutation, 900 observations were used to train the quantile neural network models, 200 were used to compute non-conformity scores, and 189 were used to test the out of sample prediction interval coverage for a 90\% prediction interval. All models were computed with the qrnn R package with one hidden layer with three nodes \citep{qrnn_package}. We also compared the conformal quantile hyperrectangle results to the na\"ive Bonferroni quantile model that used the trained Bonferroni adjusted quantile models without a conformal adjustment as well as marginal conformalized quantile regression that used a Bonferroni correction with a univariate conformal adjustment. That is, the marginal coverage rates were set to 95\% and we looked at how often they both correctly captured the blood pressure. One without a conformal adjustment and one with a conformal adjustment. Those comparisons can be found in~\cref{tab:quant_bp_data}. Coverage is the average coverage for the out of sample data over all permutations, volume is the average volume. Len i and Marg i represent the length of the interval in the i-th dimension and the marginal coverage in the i-th dimension, respectively.

\begin{table}[ht]
\begin{center}
    \begin{tabular}{|c|c|c|c|c|c|}
        \hline
         Approach &  Coverage & Marg 1 & Marg 2  & Len 1 & Len 2 \\
         \hline
        Quantile Conformal Hyperrectangles & 0.899 & 0.941 & 0.935 & 45.104 & 42.573 \\
        Na\"ive Bonferroni& 0.859 & 0.914 & 0.917 & 42.685 & 40.085 \\
        Quantile Conformal Bonferroni & 0.915 & 0.947 & 0.949 & 48.820 & 46.240 \\
         \hline
    \end{tabular}
    %\captionsetup{justification=centering}
    \caption{Blood Pressure Data Results}
    \label{tab:quant_bp_data}
\end{center}
\end{table}

It is clear from the data analysis that, although the conformal quantile hyperrectangle approach gives longer sides and a larger volume than the na\"ive quantile neural network approach with a Bonferroni adjustment, the coverage is much closer to the desired 90\%. In fact, the na\"ive method undercovers while the conformal Bonferroni correction overcovers. Looking at a boxplot of the overall coverage comparisons~\cref{fig:boxplot_comp}, it is easy to see that the na\"ive approach frequently undercovers out of sample predictions while the conformal Bonferroni correction overcovers. The conformal quantile hyperrectangular approach is the only one that achieves the desired 90\% nominal coverage on average without being overly conservative.

\begin{figure}[ht]
    \centering
\includegraphics[scale = 0.7]{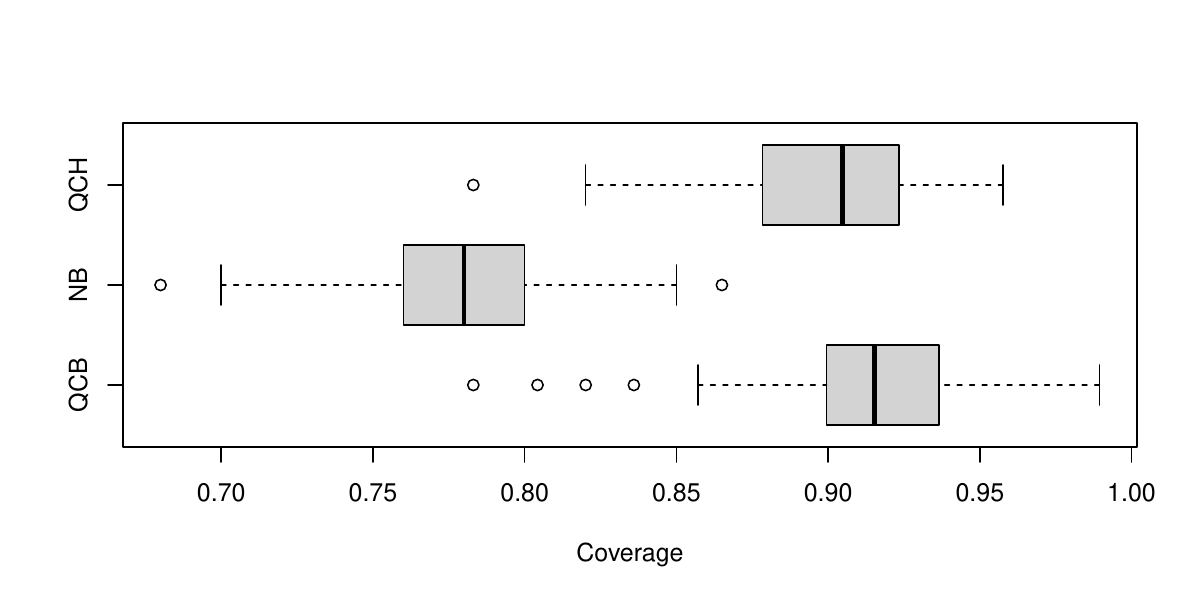}
\caption{Comparison of overall coverage. QCH represents the quantile conformal hyperrectangle approach, NB represents the na\"ive Bonferroni approach, and QCB represents the quantile conformal bonferroni approach.}\label{fig:boxplot_comp}
\end{figure}

\section{Conclusion}\label{sec:conclusion}

The conformal rectangular approach introduced here is a new way of constructing multivariate prediction intervals that do not require a computationally expensive grid search. It controls the miscoverage rate in finite samples, while allowing for heteroskedastic or homoskedastic errors. It also helps maintain balance between the dimensions of the response. Though, we may lose some precision in the model fitting process because we require the data to be split, it's clear from the real data analysis that slight loss in model precision is worth the coverage guarantees. 

More work can be done exploring other ways to compare the dimensions. Conditional variance or standard deviation also seem like reasonable measures of comparison. Another direction is to explore the use of Cartesian product of prediction sets, for example, in the case of bimodal marginal distributions. Also, based on the simulations in~\cref{sec:simulation}, there may also be weaker assumptions to guarantee balance that can be further explored.

\acks{This project was partially funded by National Institutes of Health Predoctoral Training Grant T32 HL 144461.}

\newpage
\bibliography{ref}

\section{Appendix}\label{sec:appendix}

\noindent The next two lemmas come from \cite{romanocqr} Similar results can be found in \cite{conformal_shift, split_conformal_lei_2016, conformal_book}. 

\begin{lemma} \label{lemma:quantiles_exchangeability} (Quantiles and exchangeability). Suppose $Z_1, \ldots Z_n$ are exchangeable random variables.

\noindent For any $\alpha \in (0, 1)$,
\[
\mathbb{P}(Z_n \leq \hat{Q}_n(\alpha)) \geq \alpha,
\]

\noindent where $\hat{Q}_n$ is the empirical quantile function, $\hat{Q}_n(\alpha) = Z_{(\lceil\alpha n\rceil)}$.

\noindent Moreover, if the random variables $Z_1, \ldots, Z_n$ are almost surely distinct, then
\[
\mathbb{P}(Z_n \leq \hat{Q}_n(\alpha)) \leq \alpha + \frac{1}{n}.
\]

\end{lemma}

\begin{lemma} \label{lemma:inflation_quantiles} (Inflation of Quantiles). Suppose $Z_1, \ldots Z_n$ are exchangeable random variables.

\noindent For any $\alpha \in (0, 1)$,
\[
\mathbb{P}(Z_{n+1} \leq \hat{Q}_n((1 + \frac{1}{n})\alpha)) \geq \alpha.
\]
\noindent Moreover, if the random variables $Z_1, \ldots, Z_n$ are almost surely distinct, then
\[
\mathbb{P}(Z_{n+1} \leq \hat{Q}_n((1 + \frac{1}{n})\alpha)) \leq \alpha + \frac{1}{n + 1}.
\]
    
\end{lemma}

\noindent The proof of~\cref{thm:rec_validity}
\begin{proof}\label{prf:rec_validity}
The validity of the proposed conformal hyperrectangles can be seen as follows. For CHR, let $n_{cal}$ denote $n_{cal2}$ and $\mathcal{I}_{cal}$ denote $\mathcal{I}_{cal2}$.
Call the ratios for converting between dimensions, $R_i$, $i = 1, \ldots, p$. If they depend on the covariates, call them $R_i(\vect{X}_{n+1})$, $i = 1, \ldots, p$.

%\begin{equation}
% \vect{Y}_{n+1} \in \hat{\vect{C}}(\vect{X}_{n+1}) \iff W_{n+1} \leq k    \label{eq:coverage.rate.claim}
%\end{equation}
\noindent Denote $k$ as the smallest $\lceil(n_{cal} + 1)(1 - \alpha)\rceil$ value in $\{W_i; i \in \mathcal{I}_{cal}\}$. 

\noindent Note that $W_{n+1} \leq k$ is equivalent to
\[
\iff \max \{ \hat{q}_{j, \lo}(\vect{X}_{n+1}) - Y_{j, n+1}, Y_{j, n+1} - \hat{q}_{j, \hi}(\vect{X}_{n+1}) \} \leq k \times R_j^{-1}(\vect{X}_{n+1}), \ \forall j \in \{1, \ldots, p\}
\]
\[
\iff Y_{j, n+1} \in [ \hat{q}_{j, \lo}(\vect{X}_{n+1}) - k \times R_j^{-1}(\vect{X}_{n+1}), \hat{q}_{j, \hi}(\vect{X}_{n+1}) + k \times R_j^{-1}(\vect{X}_{n+1})  ], \ \forall j \in \{1, \ldots, p\}
\]
\[
\iff \vect{Y}_{n+1} \in \hat{\vect{C}}(\vect{X}_{n+1}).
\]
So we have that,
\[
\vect{Y}_{n+1} \in \hat{\vect{C}}(\vect{X}_{n+1}) \iff W_{n+1} \leq k.
\]
Which tells us,
\[
\mathbb{P}(\vect{Y}_{n+1} \in \hat{\vect{C}}(\vect{X}_{n+1})) = \mathbb{P}(W_{n+1} \leq k).
\]
Because the original pairs of $(\vect{Y}_i, \vect{X}_i)$ are  exchangeable, so are the scores $W_i$ for $i \in \mathcal{I}_{cal}$ and $i = n + 1$. So, by~\cref{lemma:inflation_quantiles}, 
\[
\mathbb{P}(\vect{Y}_{n+1} \in \hat{\vect{C}}(\vect{X}_{n+1})) \geq 1 - \alpha,
\]
\noindent and, under the assumption that the $W_i$'s are almost surely distinct, 
\[
\mathbb{P}(\vect{Y}_{n+1} \in \hat{\vect{C}}(\vect{X}_{n+1})) \leq 1 - \alpha + \frac{1}{n_{cal} + 1}.
\]
\end{proof}

\begin{lemma}  \citep[Lemma 2.2]{van1998asymptotic}\label{lemma:lipschitz_weak_convergence}
    \[
    X_n \overset{d}{\to} X \] 
    holds  
    iff, for all bounded,  Lipschitz functions $g$,
\[\lim_{n\to\infty}\mathbb{E}(g(X_n)) = \mathbb{E}(g(X))
    \]
\end{lemma}

\begin{lemma}\label{lemma:convergence_distribution}
    For a compact set, $\mathcal{C}$, assume that
    \begin{equation}\label{eq:uniform_consistency}
    \sup_{x \in \mathcal{C}} |\hat{\mu}_j(\vect{x}) - \mu_j(\vect{x})| \overset{a.s.}{\to} 0 \text{ as } n_{tr} \to \infty \ \forall j.
    \end{equation}

    \noindent Then we have that as $n_{tr}$ and $n_{cal1}$ $\to \infty$, the empirical probability measure
    \[
    \frac{1}{n_{cal1}} \sum_{i = 1}^{n_{cal1}} \delta(Y_{i,j} - \hat{\mu}_j(\vect{X}_i)) \overset{d}{\to} Y_j - \mu_j(\vect{X}), \quad \forall j,
    \]
    where $\delta(\cdot)$ denotes the Dirac delta probability measure at the vector enclosed within the parentheses. 
\end{lemma}

\begin{proof}
    Let $f$ be a bounded  Lipschitz function. Then we have for some constant $K$ that,
    \[
    \Big|\frac{1}{n_{cal1}} \sum_{i = 1}^{n_{cal1}} f(Y_{i,j} - \hat{\mu}_j(\vect{X}_i)) - \frac{1}{n_{cal1}} \sum_{i = 1}^{n_{cal1}} f(Y_{i,j} - \mu_j(\vect{X}_i))\Big| \leq K \max_{i = 1, \ldots, n_{cal1}}|\hat{\mu}_j(\vect{X}_i) - \mu(\vect{X}_i)| \overset{a.s.}{\to} 0,
    \]
  as $n_{tr} \to \infty$, $\forall j,$
  by \eqref{eq:uniform_consistency}.
  % as $n_{tr} \to \infty$,
  %   \begin{equation} \label{eq:lipschitz_unif_conssitency}
  %   K \max_{i = 1, \ldots, n_{cal1}}|\hat{\mu}_j(\vect{X}_i) - \mu(\vect{X}_i)| \overset{a.s.}{\to} 0 \ \forall j.
  %   \end{equation}
    Now, by the strong law of large numbers, 
    \[
    \frac{1}{n_{cal1}} \sum_{i = 1}^{n_{cal1}} f(Y_{i,j} - \mu_j(\vect{X}_i)) \to \mathbb{E}(f(Y_{j} - \mu_j(\vect{X}))) \text{ almost surely, as } n_{cal1} \to \infty,
    \]
    % and
    % \[
    % \frac{1}{n_{cal1}} \sum_{i = 1}^{n_{cal1}} f(Y_{i,j} - \hat{\mu}_j(\vect{X}_i)) \to \mathbb{E}(f(Y_{j} - \hat{\mu}_j(\vect{X}))) \text{ as } n_{cal1} \to \infty.
    % \]
    % So, by the above and \eqref{eq:lipschitz_unif_conssitency}, 
    hence,
    \[
    \frac{1}{n_{cal1}}\sum_{i = 1}^{n_{cal1}} f(Y_{i,j} - \hat{\mu}_j(\vect{X}_i))  \to \mathbb{E}(f(Y_{j} - \mu_j(\vect{X}))) \text{ as } n_{cal1} \text{ and } n_{tr} \to \infty.
    \]
    Then by~\cref{lemma:lipschitz_weak_convergence}, $\frac{1}{n_{cal1}} \sum_{i = 1}^{n_{cal1}} \delta(Y_{i,j} - \hat{\mu}_j(\vect{X}_i)\overset{d}{\to} Y_j - \mu_j(\vect{X})$.
\end{proof}

\begin{lemma} \citep[Lemma 21.2]{van1998asymptotic} \label{lemma:continuity_points}
    Let $X$ be a random variable with distribution $F$. Define the quantile function of a cumulative distribution F as $F^{-1}(p) = \inf\{x:F(x) \geq p \}$. Then, for any $0 < p < 1$, $F_n^{-1}(p) \to F^{-1}(p)$ at all continuity points p of $F^{-1}$ if and only if $X_n \overset{d}{\to} X$. 
\end{lemma}

\noindent The proof of~\cref{thm:quant_rec_bal}
\begin{proof} \label{prf:quant_rec_bal}
    Denote the oracle conditional quantiles as $[q_{j,\lo}(\vect{x}), q_{j,\hi}(\vect{x})]$ and the consistent estimators as $[\hat{q}_{j,\lo}(\vect{x}), \hat{q}_{j,\hi}(\vect{x})]$. Let the first dimension be the reference dimension, that is, $c_1 =1,  \hat{c}_1 (\vect{x})\equiv 1$.

    \noindent By consistency, we know that, uniformly in $\vect{x}\in \mathcal{C},  \ j = 1, \ldots, p$, $P(| q_{j,\lo}(\vect{x}) - \hat{q}_{j,\lo}(\vect{x})| > \delta_{\lo}) \to 0$ for every $\delta_{\lo} > 0$ as $n_{tr} \to \infty$ and  $P(| q_{j,\hi}(\vect{x}) - \hat{q}_{j,\hi}(\vect{x})| > \delta_{\hi}) \to 0$ for every $\delta_{\hi} > 0$ as $n_{tr} \to \infty$.
Hence, for $j=2,\ldots, p$, \[
\hat{c}_j(\vect{x})=\frac{\hat{c}_j(\vect{x})}{\hat{c}_1(\vect{x})} = \frac{\hat{q}_{1,\hi}(\vect{x}) - \hat{q}_{1,\lo}(\vect{x})}{\hat{q}_{j,\hi}(\vect{x}) - \hat{q}_{j,\lo}(\vect{x})}
\]
% which 
% \noindent By consistency of the conditional qunatile estimators and because $c_1>0$, 
% \[
% \frac{\hat{c}_j}{\hat{c}_1} \overset{P}{\to} \frac{c_j}{c_1} = \frac{c_j}{1},
% \]
converges  in probability to $c_j/c_1=c_j$ uniformly in $\vect{x}\in \mathcal{C},  \ j = 2, \ldots, p.$ Hence, for all $\delta_j>0, j=2,\ldots, p$ and uniformly in $\vect{x}$,  $P(|\hat{c}_j(\vect{x})-c_j|> \delta_j)\to 0$, in probability, as $n_{tr}\to\infty$.
% Replacing $\hat{c}_1$ with $1$, we then have that,
% \[
% \hat{c}_j \overset{P}{\to} c_j.
% \]
% \[
% P(|\hat{c}_j - c_j| > \delta_j) \to 0, \ \text{as} \ n_{tr} \to \infty, \ \forall \ \delta_j > 0, \ \ j = 2, \ldots, p.
% \]

\noindent Define the symmetric set difference, $A \Delta B = (A \cap B^{C})\cup (A^{C} \cap B)$.
    
\noindent  For each $j = 1, \ldots, p$, a large enough $n_{tr}$, and a constant $\ell$, by consistency we have that
    \begin{align*}
    P(Y_j &\in [q_{j,\lo}(\vect{x}) - \frac{1}{c_j} \ell, q_{j,\hi}(\vect{x}) + \frac{1}{c_j} \ell] \Delta [\hat{q}_{j,\lo}(\vect{x}) -  \frac{1}{\hat{c}_j(\vect{x})} \ell, \hat{q}_{j,\hi}(\vect{x}) + \frac{1}{\hat{c}_j(\vect{x})}\ell ]) \\
    \leq P(Y_j &\in [{q}_{j,\lo}(\vect{x}) - \delta_{\lo} - \frac{1}{c_j - \delta_j}\ell, {q}_{j, \lo}(\vect{x}) + \delta_{\lo} - \frac{1}{c_j +\delta_j}\ell] \\ %the denominator of the fraction was chosen to subtract off the largest and smallest, so the interval is at its widest
    &\cup [{q}_{j,\hi}(\vect{x}) - \delta_{\hi} + \frac{1}{c_j + \delta_j}\ell, {q}_{j,\hi}(\vect{x}) + \delta_{\hi}  +  \frac{1}{c_j - \delta_j}\ell])
    \end{align*}
Because $\delta_{\lo}$, $\delta_{\hi}$, and $\delta_j$ are arbitrarily small positive numbers, as $n_{tr} \to \infty$
    \[
    P(Y_j \in [q_{j,\lo}(\vect{x}) - \frac{1}{c_j} \ell, q_{j,\hi}(\vect{x}) + \frac{1}{c_j} \ell] \Delta [\hat{q}_{j,\lo}(\vect{x}) -  \frac{1}{\hat{c}_j(\vect{x})} \ell, \hat{q}_{j,\hi}(\vect{x}) + \frac{1}{\hat{c}_j(\vect{x})}\ell ]) \overset{P}{\to} 0.
    \]
By the ICED assumption, the scaled oracle set has the same miscoverage, $\gamma \in (0,1)$, in each dimension. So using the ICED and the above, 
\begin{equation*}
\eta_j = P(Y_j \notin [\hat{q}_{j,\lo}(\vect{x}) - \frac{1}{\hat{c}_j(\vect{x})} \ell, \hat{q}_{j,\hi}(\vect{x}) +  \frac{1}{\hat{c}_j(\vect{x})} \ell]) \overset{P}{\to} \gamma, \ j = 1, \ldots, p.
\end{equation*}
So, 
\[
\max_{j}\{|\eta_j - \frac{1}{p}\sum\limits_{i = 1}^p \eta_i|\} \overset{P}{\to} 0.
\]
Which is the definition for asymptotic balance. Because $\ell$ is a general constant, we can let it be the conformal adjustment, so conformal quantile hyperrectangular regression achieves asymptotic balance under the listed assumptions.
\end{proof}

\noindent The proof of~\cref{thm:point_rec_bal}
\begin{proof} \label{prf:point_rec_bal}
Note the CQ assumption implies the validity of the ICED assumption. 
By the ICED assumption we have that,
\[
|c_j(Y_j - \mu_j(\vect{X})| \overset{\mathcal{D}}{=} G, \ j = 1, \ldots, p,
\]
and by~\cref{lemma:convergence_distribution}, the empirical  probability distribution induced by $c_j\{Y_{i,j}-\hat{\mu}_j(\vect{X}_i)\}, i=1,\ldots, n_{cal1}$, denoted by $G_{n,j}$,  converges in distribution to $G$, as both $n_{tr}$ and $n_{cal1}\to \infty$, for $j = 1, \ldots, p$.
% \[
% |c_j(Y_j - \hat{\mu}_j(\vect{X}))| \overset{d}{\to} G, \ j = 1, \ldots, p.
% \]
\noindent 
% Because each dimension will have a slightly different scaled empirical distribution, denote it as $G_{n, j}$. 
Let $c_j\hat{v}_j = G_{n, j}^{-1}(\alpha)$ and $c_jv_j = G^{-1}(\alpha)$, where $G_{n, j}^{-1}(\alpha)$ and $G^{-1}(\alpha)$ are the empirical and true $\alpha$ quantiles respectively. Note that $\hat{v}_j$ and $v_j$ are the observed initial conformal expansions and the oracle initial conformal expansions. By assumption CQ and~\cref{lemma:continuity_points},
\begin{equation}\label{eq:balance_scale_point_estimate}
\hat{v}_j = \frac{G_{n, j}^{-1}(\alpha)}{c_j} \to \frac{G^{-1}(\alpha)}{c_j} = v_j \text{ as } n_{tr} \mbox{ and } n_{cal1} \to \infty, \ j = 1, \ldots, p.
\end{equation}
Next, let the first dimension be the reference, that is, $c_1 = 1, \hat{c}_1 \equiv 1$. Then, our estimates of $\hat{c}_j$ are $\frac{\hat{v}_1}{\hat{v}_j}$. By \eqref{eq:balance_scale_point_estimate}, for $j = 2, \ldots p$,
\begin{equation*}
\hat{c}_j = \frac{\hat{c}_j}{\hat{c}_1} = \frac{\hat{v}_1}{\hat{v}_j} \to \frac{v_1}{v_j} = \frac{c_j}{c_1} = c_j \text{ as } n_{cal1} \to \infty.
\end{equation*}
Define $\zeta > 0$, $\epsilon > 0$, and $\delta > 0$. For a large enough $n_{tr}$ and $n_{cal1}$, and some constant, $\ell$, by the ICED assumption and the consistency of $\hat{\mu}$, we have
\[
P(Y_j \in  [\mu_j(\vect{x}) - \frac{1}{c_j}(v_1 + \ell), \mu_j(\vect{x}) + \frac{1}{c_j}(v_1 + \ell)] \Delta [\hat{\mu}_j(\vect{x}) - \frac{1}{\hat{c}_j}(\hat{v}_1 + \ell), \hat{\mu}_j(\vect{x}) + \frac{1}{\hat{c}_j}(\hat{v}_1 + \ell)])
\]
\begin{align*}
\leq  P(Y_j &\in[\mu_j(\vect{x}) - \zeta - \frac{1}{c_j - \delta}(v_1 + \epsilon + \ell), \mu_j(\vect{x}) + \zeta - \frac{1}{c_j - \delta}(v_1 - \epsilon + \ell)] \\&\cup [\mu_j(\vect{x}) - \zeta + \frac{1}{c_j - \delta}(v_1 - \epsilon + \ell), \mu_j(\vect{x}) + \zeta+ \frac{1}{c_j - \delta}(v_1 + \epsilon + \ell)])
\end{align*}

\noindent Because $\zeta$, $\delta$, and $\epsilon$ are arbitrarily small numbers, as $n_{tr} \to \infty$ and $n_{cal1} \to \infty$,
\[
P(Y_j \in  [\mu_j(\vect{x}) - \frac{1}{c_j}(v_1 + \ell), \mu_j(\vect{x}) + \frac{1}{c_j}(v_1 + \ell)] \Delta [\hat{\mu}_j(\vect{x}) - \frac{1}{\hat{c}_j}(\hat{v}_1 + \ell), \hat{\mu}_j(\vect{x}) + \frac{1}{\hat{c}_j}(\hat{v}_1 + \ell)]) \overset{P}{\to} 0.
\]
By the ICED assumption, the scaled oracle set has the same miscoverage, $\gamma \in (0, 1)$, in each dimension. So, using the ICED and the above, as $n_{tr} \to \infty$ and $n_{cal1} \to \infty$,
\[
\eta_j = P(Y_j \notin [\hat{\mu}_j(\vect{x}) - \frac{1}{\hat{c}_j}(\hat{v}_1 + \ell), \hat{\mu}_j(\vect{x}) + \frac{1}{\hat{c}_j}(\hat{v}_1 + \ell)] \overset{P}{\to} \gamma, \ j = 1, \ldots, p. 
\]
So, 
\[
\max_{j}\{|\eta_j - \frac{1}{p}\sum\limits_{i = 1}^p \eta_i|\} \overset{P}{\to} 0.
\]
Which is the definition for asymptotic balance. Because $\ell$ is a general constant, we can set it to be the final conformal adjustment, so CHR achieves asymptotic balance under the listed assumptions.
\end{proof}
\end{document}